\documentclass[letter,final]{IEEEtran}

\usepackage{cite}
\usepackage[cmex10]{amsmath}
\usepackage{amssymb}
\usepackage{mathrsfs}
\usepackage{amsfonts}
\usepackage{amsthm}
\usepackage{setspace}
\usepackage{graphicx,color,tikz,pgfplots}
\usepackage{caption}
\usepackage{array}
\usepackage{makeidx}
\usepackage{float}
\usepackage{pxfonts}
\usepackage{algorithm}
\usepackage{algorithmic}
\usepackage{url}
\usepackage{empheq}
\usepackage{lmodern}
\usepackage{multirow}
\usepackage{colortbl}
\usepackage{hhline}
\newcommand*{\vecteur}[1]{\textbf{#1}}

\newcommand*{\bsym}[1]{\boldsymbol{#1}}

\renewcommand{\vec}{\text{vec}}
\newcommand*{\bm}[1]{{\color{black} #1}}

\theoremstyle{definition}
\newtheorem{theorem}{Theorem}[section]

\theoremstyle{plain}
\newtheorem{lemma}{Lemma}[section]

\theoremstyle{plain}
\newtheorem*{remark}{Remark}

\begin{document}

\title{Asymptotic Performance of Complex $M$-estimators for Multivariate Location and Scatter Estimation}
%
%
%

\author{Bruno~M\'eriaux,~\IEEEmembership{Student Member,~IEEE,}
        Chengfang~Ren,~\IEEEmembership{Member,~IEEE,}
        Mohammed~Nabil~El~Korso,~\IEEEmembership{}
        Arnaud~Breloy,~\IEEEmembership{Member,~IEEE,}
        and~Philippe~Forster~\IEEEmembership{}
\thanks{\bm{B. M\'eriaux and C. Ren are with SONDRA, CentraleSup\'elec, France. M.N. El Korso and A. Breloy are with LEME, Paris-Nanterre University, France. P. Forster is with SATIE, Paris-Nanterre University, France.}}
\thanks{This work is financed by the Direction G\'en\'erale de l'Armement as well as the ANR ASTRID referenced ANR-17-ASTR-0015.}%
}

\maketitle
\begin{abstract}
The joint estimation of means and scatter matrices is often a core problem in multivariate analysis. In order to overcome robustness issues, such as outliers from Gaussian assumption, $M$-estimators are now preferred to the traditional sample mean and sample covariance matrix. These estimators are well established and studied in the real case since the seventies. Their extension to the complex case has drawn recent interest. In this letter, we derive the asymptotic performance of complex $M$-estimators for multivariate location and scatter matrix estimation.
\end{abstract}

\IEEEpeerreviewmaketitle
\section{Introduction}
\label{sec::intro}
\IEEEPARstart{S}{everal} classical methods in multivariate analysis require the estimation of means and scatter matrices from \bm{collected observations \cite{FPOP17,Fess96,ZKCM12,OEKBLP17,ZEKP17}. In pratice, the Sample Mean and the Sample Covariance Matrix are classically used} in such procedures. Indeed, they \bm{coincide with} the Maximum Likelihood Estimators (MLE) for multivariate Gaussian data. However, they are neither robust to deviation from Gaussianity assumption nor to the presence of outliers, which could lead to a dramatic performance loss. To overcome these problems, several approaches have been proposed in the literature\bm{\cite{HRVV15,MY17}} improving the estimator's behavior under contamination through diverse criteria such as the breakdown point, the contamination bias, the finite-sample efficiency while preserving the computationally feasibility for high dimension. Among the most frequently used robust affine equivariant estimators \cite{Mar76,Rou85,Rou84,HD09,Sta81,Don82,RL87,Lop89,Yoh87,TT00}, we focus on the $M$-estimators, which have been first introduced within the framework of elliptical distributions \cite{Fra04}. The latters encompass a large number of classical distributions as for instance the Gaussian one but also non-Gaussian heavy-tailed distributions such as the $t$-, $K$- and $W$-distributions \cite{FKW90}. One of the main interests of the real $M$-estimators is to possess, under mild conditions, good asymptotic properties, namely weak consistency and asymptotic normality over the whole class of elliptical distributions\bm{\cite{Mar76,Tyl82,BB08,Ars04}}. Their extension to the complex case has drawn recent interest\bm{\cite{OK03}}, notably with the class of Complex Elliptically Symmetric (CES) distributions\bm{\cite{GG13}}. In the context of known mean, their asymptotic properties are established in \cite{MPFO13,OTKP12}. In this letter, we extend this result for multivariate location and scatter matrix $M$-estimates under CES distributed data. The achieved outcome is analogous to its real-counterpart \cite{BB08} and may be used for deriving the performance of adaptive processes in non-zero mean non Gaussian distributed observations \cite{FPOP17}.

In the following, the notation $\overset{d}{=}$, $\overset{d}{\rightarrow}$ and $\overset{\mathbb{P}}{\rightarrow}$ indicate respectively equality in distribution, convergence in law and in probability. The symbol $\perp$ refers to statistical independence. The operator $\vec{}(\vecteur{A})$ stacks all columns of $\vecteur{A}$, designated by $\left(\vecteur{a}_1,\ldots,\vecteur{a}_i,\ldots\right)$ into a vector. The operator $\otimes$ refers to Kronecker product. The notation $\mathcal{G}\mathbb{C}\mathcal{N}\left(\vecteur{0}, \bsym{\Sigma},\bsym{\Omega}\right)$ refers to the zero mean non-circular complex Gaussian distribution, where $\bsym{\Sigma}$ (respectively $\bsym{\Omega}$) denotes the covariance matrix (respectively pseudo-covariance matrix) \cite{Van95,Del14}.
\section{Problem setup}
\label{sec::setup}
Let {\small $\vecteur{Z}_N \triangleq\left(\vecteur{z}_1,\ldots,\vecteur{z}_N\right)$} be $N$ i.i.d samples of $m$-dimensional vectors, following a complex elliptical distribution, which is denoted by $\vecteur{z}_n\sim\mathbb{C}\mathcal{ES}_m\left(\vecteur{t}_e,\bsym{\Lambda},g_{\vecteur{z}}\right)$ and whose p.d.f. is proportional to
{\footnotesize \begin{equation}
p_{\vecteur{z}}\left(\vecteur{z}_n;\vecteur{t}_e,\bsym{\Lambda},g_{\vecteur{z}}\right)\propto \text{det}(\bsym{\Lambda})^{-1}g_{\vecteur{z}}\left(\left(\vecteur{z}_n-\vecteur{t}_e\right)^H\bsym{\Lambda}^{-1}\left(\vecteur{z}_n-\vecteur{t}_e\right)\right).
\end{equation}}
The vector $\vecteur{t}_e$ is the location parameter, $\bsym{\Lambda}$ denotes the scatter matrix and the function $g_{\vecteur{z}}$ is the density generator \cite{Kel70}. We aim to estimate jointly the location parameter and the scatter matrix from the observations. The complex joint $M$-estimators $\big(\mathbf{\hat{t}}_N,\mathbf{\widehat{M}}_N\big)$ are solutions of the system $\mathcal{S}\mathit{ys}_N\left(\vecteur{Z}_N,u_1,u_2\right)$ \cite{OK03}:
{\footnotesize
\begin{empheq}[left=\empheqlbrace]{align}
& \dfrac{1}{N}\sum\limits_{n=1}^N u_1\left(d\left(\vecteur{z}_n,\vecteur{t};\vecteur{M}\right)\right)\left(\vecteur{z}_n-\vecteur{t}\right) = \vecteur{0} \label{eq:moy_emp_comp}\\
& \underbrace{\dfrac{1}{N}\sum\limits_{n=1}^N u_2\left(d^2\left(\vecteur{z}_n,\vecteur{t};\vecteur{M}\right)\right)\left(\vecteur{z}_n-\vecteur{t}\right)\left(\vecteur{z}_n-\vecteur{t}\right)^H}_{\mathcal{H}\left(\vecteur{Z}_N,\vecteur{t},\vecteur{M}\right)} = \vecteur{M} \label{eq:scatter_emp_comp}
\end{empheq}}
with {\footnotesize $d^2\left(\vecteur{z},\vecteur{t};\vecteur{M}\right) = \left(\vecteur{z}-\vecteur{t}\right)^H\vecteur{M}^{-1}\left(\vecteur{z}-\vecteur{t}\right)$}. Let us consider {\small $\big(\vecteur{t}_e,\vecteur{M}_e\big)$} a solution related to the system $\mathcal{S}\mathit{ys}_{\infty}\left(\vecteur{z}_1,u_1,u_2\right)$:
{\footnotesize
\begin{empheq}[left=\empheqlbrace]{align}
& \mathbb{E}\left[ u_1\left(d\left(\vecteur{z}_1,\vecteur{t};\vecteur{M}\right)\right)\left(\vecteur{z}_1-\vecteur{t}\right)\right] = \vecteur{0} \\
\label{scat_conv}
&\mathbb{E}\left[ u_2\left(d^2\left(\vecteur{z}_1,\vecteur{t};\vecteur{M}\right)\right)\left(\vecteur{z}_1-\vecteur{t}\right)\left(\vecteur{z}_1-\vecteur{t}\right)^H\right]\triangleq \mathcal{H}_ {\infty}\left(\vecteur{t},\vecteur{M}\right) = \vecteur{M}
\end{empheq}}%
The functions $u_1(\cdot)$ and $u_2(\cdot)$ verify the conditions given in \cite{Mar76} \bm{for the real case (for the special case, where $u_1(s) = u_2(s^2)$, \cite{KT91} provides more general conditions). The proofs of Lemmas 1 and 2 and Theorems 1-3 of \cite{Mar76}, addressing the existence and uniqueness of $M$-estimates, are transposable to the complex field, by following the same methodology as in \cite{Mar76}. The derivations of these proofs require the same conditions on $u_1(\cdot)$ and $u_2(\cdot)$ as the ones needed in the real case. Thus, this ensures} the existence of {\small $\big(\mathbf{\hat{t}}_N,\mathbf{\widehat{M}}_N\big)$} and {\small $\big(\vecteur{t}_e,\vecteur{M}_e\big)$} as well as the uniqueness of {\small $\big(\vecteur{t}_e,\vecteur{M}_e\big)$}. In this letter, we derive the statistical performance of the complex joint $M$-estimators {\small $\big(\mathbf{\hat{t}}_N,\mathbf{\widehat{M}}_N\big)$}, namely consistency and asymptotic distribution.
\section{Consistency of the joint $M$-estimator}
\label{sec::consistency}

Let {\small $\big(\mathbf{\hat{t}}_N,\mathbf{\widehat{M}}_N\big)$} be a solution of $\mathcal{S}\mathit{ys}_N\left(\vecteur{Z}_N,u_1,u_2\right)$ and $\big(\vecteur{t}_e,\vecteur{M}_e\big)$ be the solution of the system $\mathcal{S}\mathit{ys}_{\infty}\left(\vecteur{z}_1,u_1,u_2\right)$.
\begin{theorem}
\label{small_theorem}
The complex joint $M$-estimators are consistent, i.e. \vspace*{-0.3cm}
{\footnotesize 
\begin{align}
\big(\mathbf{\hat{t}}_N,\mathbf{\widehat{M}}_N\big)\overset{\mathbb{P}}{\rightarrow} \big(\vecteur{t}_e,\vecteur{M}_e\big).
\end{align}}%
with $\vecteur{M}_e = \sigma^{-1}\bsym{\Lambda}$ in which $\sigma$ is the solution of $\mathbb{E}\left[\psi_2\left(\sigma\vert\bsym{\zeta}\vert^2\right)\right] = m$, for $\bsym{\zeta}\sim\mathbb{C}\mathcal{ES}_m\left(\vecteur{0}_m,\vecteur{I}_m,g_{\vecteur{z}}\right)$ and $\psi_2(s) = su_2(s)$.
\end{theorem}
\begin{proof}
First, let us define $\bsym{\theta}^T \hspace*{-.2cm} = \left[\vecteur{t}^T\hspace*{-.2cm},\vec{}\left(\vecteur{M}\right)^T\right]$ and the function {\footnotesize
$\bsym{\Psi}_N\left(\bsym{\theta}\right) = \left[\hspace*{-0.2cm}\begin{array}{l}
\bsym{\Psi}_{1,N}\left(\bsym{\theta}\right)= \dfrac{1}{N}\sum\limits_{n=1}^N u_1\left(d\left(\vecteur{z}_n,\vecteur{t};\vecteur{M}\right)\right)\left(\vecteur{z}_n-\vecteur{t}\right) \\ 
\bsym{\Psi}_{2,N}\left(\bsym{\theta}\right)= \vec{}\left(\mathcal{H}\left(\vecteur{Z}_N,\vecteur{t},\vecteur{M}\right) - \vecteur{M}\right)
\end{array}\hspace*{-0.2cm}\right] $}.\\
\linebreak
The Strong Law of Large Numbers (SLLN) gives
{\small \begin{align}
\forall\,\bsym{\theta}\in\bsym{\Theta}\quad \bsym{\Psi}_N\left(\bsym{\theta}\right) \overset{\mathbb{P}}{\rightarrow}\bsym{\Psi}\left(\bsym{\theta}\right)
\end{align}}%
with {\footnotesize $\forall\,\bsym{\theta}\in\bsym{\Theta},\;\bsym{\Psi}\left(\bsym{\theta}\right) = \left[\begin{array}{l}
\bsym{\Psi}_{1}\left(\bsym{\theta}\right)= \mathbb{E}\left[ u_1\left(d\left(\vecteur{z}_1,\vecteur{t};\vecteur{M}\right)\right)\left(\vecteur{z}_1-\vecteur{t}\right)\right] \\ 
\bsym{\Psi}_{2}\left(\bsym{\theta}\right)= \vec{}\left(\mathcal{H}_ {\infty}\left(\vecteur{t},\vecteur{M}\right) - \vecteur{M}\right)
\end{array}\right] $}.\\
\linebreak
According to the Theorem 5.9 \cite[Chap. 5]{Van00} and uniqueness of solution, we can show that any {\small $\widehat{\bsym{\theta}}_N^T = \left[\mathbf{\hat{t}}_N^T,\vec{}\big(\mathbf{\widehat{M}}_N\big)^T\right]$} solution of {\small $\bsym{\Psi}_N\left(\widehat{\bsym{\theta}}_N\right) = \vecteur{0}$} converges in probability to {\small $\bsym{\theta}_e^T = \left[\vecteur{t}_e^T,\vec{}\left(\vecteur{M}_e\right)^T\right]$} solution of {\small $\bsym{\Psi}\left(\bsym{\theta}_e\right) = \vecteur{0}$}, \bm{yielding} the intended outcome. Furthermore, the matrix $\vecteur{M}_e$ is proportional to $\bsym{\Lambda}$ through a scale factor $\sigma^{-1}$ \cite[Chap. 6]{MMY06}. Multiplying \eqref{scat_conv} by $\vecteur{M}^{-1}$ and taking the trace yields $\mathbb{E}\left[\psi_2\left(\sigma\vert\bsym{\zeta}\vert^2\right)\right] = m$ of which $\sigma$ is the solution.
\end{proof}

\section{Asymptotic distribution of the joint $M$-estimators}
\subsection{Main theorem}
\label{sec::dist_asympt}
We consider {\small $\big(\mathbf{\hat{t}}_N,\mathbf{\widehat{M}}_N\big)$} a solution of $\mathcal{S}\mathit{ys}_N\left(\vecteur{Z}_N,u_1,u_2\right)$ as well as $\big(\vecteur{t}_e,\vecteur{M}_e\big)$ the solution of $\mathcal{S}\mathit{ys}_{\infty}\left(\vecteur{z}_1,u_1,u_2\right)$.
\begin{theorem}
\label{main_theorem}
Assuming that $s\psi_i^{'}(s)$ ($i=1,2$) are bounded and $\mathbb{E}\left[\psi_1^{'}\left(\sqrt{\sigma}\vert\bsym{\zeta}\vert\right)\right] > 0$, the asymptotic distribution of $\big(\mathbf{\hat{t}}_N,\mathbf{\widehat{M}}_N\big)$ is given by
{\small \begin{align*}
&\sqrt{N}\left(\big(\mathbf{\hat{t}}_N-\vecteur{t}_e,\textnormal{\vec{}}\left(\mathbf{\widehat{M}}_N-\vecteur{M}_e\right)\right) \overset{d}{\rightarrow}\left(\vecteur{a},\vecteur{b}\right)\text{with $\vecteur{a}\perp\vecteur{b}$ and}\\
&\vecteur{a}\sim\mathbb{C}\mathcal{N}\big(\vecteur{0}, \bsym{\Sigma}_t=\dfrac{\alpha}{\beta^2}\vecteur{M}_e\big),\,
\vecteur{b}\sim\mathcal{G}\mathbb{C}\mathcal{N}\left(\vecteur{0}, \bsym{\Sigma}_M,\bsym{\Omega}_M=\bsym{\Sigma}_M\vecteur{K}_m\right)
\end{align*}}%
where $\vecteur{K}_m$ is the commutation matrix satisfying $\vecteur{K}_m\vec{}\left(\vecteur{A}\right)=\vec{}\big(\vecteur{A}^T\big)$ \cite{MN79} and $\bsym{\Sigma}_M$ is obtained by
{\small \begin{align}
\bsym{\Sigma}_M = \sigma_1\vecteur{M}_e^T\otimes\vecteur{M}_e + \sigma_2\textnormal{\vec{}}\left(\vecteur{M}_e\right)\textnormal{\vec{}}\left(\vecteur{M}_e\right)^H,
\end{align}}%
in which {\footnotesize \begin{align*}
\hspace*{-0.15cm}\left\lbrace\begin{array}{l}
\alpha = m^{-1}\mathbb{E}\left[\psi_{1}^2(\sqrt{\sigma}\vert\bsym{\zeta}\vert)\right] \\
\beta = \mathbb{E}\left[\left(1-(2m)^{-1}\right)u_{1}(\sqrt{\sigma}\vert\bsym{\zeta}\vert)+(2m)^{-1}\psi_{1}^{'}\left(\sqrt{\sigma}\vert\bsym{\zeta}\vert\right)\right] \\
\sigma_1 = \dfrac{a_1(m+1)^2}{(a_2+m)^2},
\sigma_2 = a_2^{-2}\left[(a_1 -1) - a_1(a_2-1)\dfrac{m + (m+2)a_2}{(a_2 + m)^2}\right]\\
a_1 = \dfrac{\mathbb{E}\left[\psi_2^2\left(\sigma\vert\bsym{\zeta}\vert^2\right)\right]}{m(m+1)},\;
a_2 = \dfrac{\mathbb{E}\left[\sigma\vert\bsym{\zeta}\vert^2\psi_2^{'}\left(\sigma\vert\bsym{\zeta}\vert^2\right)\right]}{m}
\end{array}\right.
\end{align*}}%
with $\sigma$ solution of $\mathbb{E}\left[\psi_{2}\left(\sigma\vert\bsym{\zeta}\vert^2\right)\right] = m$ in which $\bsym{\zeta}\sim\mathbb{C}\mathcal{ES}_m\left(\vecteur{0}_{m},\vecteur{I}_{m},g_{\vecteur{z}}\right)$.
\end{theorem}

\subsection{Proof of Theorem \ref{main_theorem}}
The starting point of the proof is to map the complex joint $M$-estimators $\big(\mathbf{\hat{t}}_N,\mathbf{\widehat{M}}_N\big)$ into real-ones, then to study the asymptotic behavior of the latter, and finally to relate the latter to the asymptotic distribution of the complex joint $M$-estimators $\big(\mathbf{\hat{t}}_N,\mathbf{\widehat{M}}_N\big)$.
\subsubsection{Complex vector space isomorphism}
Let us first introduce \bm{functions $h:\mathbb{C}^{m}\rightarrow\mathbb{R}^{p}$ and $f:\mathbb{C}^{m\times m}\rightarrow\mathbb{R}^{p\times p}$ with $p=2m$ defined by {\small $h\left(\vecteur{a}\right) = \left(\Re\left(\vecteur{a}\right)^T,\Im\left(\vecteur{a}\right)^T\right)^T$} and} {\footnotesize \begin{align*}
f\left(\vecteur{A}\right) = \dfrac{1}{2}\begin{pmatrix}
\Re\left(\vecteur{A}\right) & -\Im\left(\vecteur{A}\right) \\ 
\Im\left(\vecteur{A}\right) & \Re\left(\vecteur{A}\right)
\end{pmatrix} 
\end{align*}}%
In addition, let $\bsym{\mathcal{P}}\in\mathbb{R}^{p\times p}$ be the matrix {\footnotesize $\bsym{\mathcal{P}} = \begin{pmatrix}
\vecteur{0}_{m\times m} & -\vecteur{I}_m \\ 
\vecteur{I}_m & \vecteur{0}_{m\times m}
\end{pmatrix}$}. Some useful properties of the previous functions are given in \cite{MPFO13}. Furthermore, we set {\small $\mathbf{\hat{t}}_{N}^{\mathbb{R}} = h\left(\mathbf{\hat{t}}_N\right)$, $\mathbf{\widehat{M}}_{N}^{\mathbb{R}} = f\left(\mathbf{\widehat{M}}_N\right)$, $\vecteur{M}_{\mathbb{R}} = f\left(\vecteur{M}_e\right)$ and $\vecteur{t}_{\mathbb{R}} = h\left(\vecteur{t}_e\right)$}. In addition, let us define $\vecteur{u}_n = h\left(\vecteur{z}_n\right)\sim\mathcal{ES}_p\left(\vecteur{t}_{e}^{u},\bsym{\Lambda}_{\mathbb{R}},g_{\vecteur{z}}\right)$ and $\vecteur{v}_n =\bsym{\mathcal{P}}\vecteur{u}_n\sim\mathcal{ES}_p\left(\vecteur{t}_{e}^{v},\bsym{\Lambda}_{\mathbb{R}},g_{\vecteur{z}}\right)$ for $\vecteur{z}_n\sim\mathbb{C}\mathcal{ES}_m\left(\vecteur{t}_e,\bsym{\Lambda},g_{\vecteur{z}}\right)$, where $\vecteur{t}_{e}^{u} = h\left(\vecteur{t}_{e}\right)$, $\vecteur{t}_{e}^{v}= \bsym{\mathcal{P}}\vecteur{t}_{e}^{u}$ and $\bsym{\Lambda}_{\mathbb{R}} = f\left(\bsym{\Lambda}\right)$. The notation $\mathcal{ES}$ refers to real elliptical distributions \cite{Fra04}. Moreover, there exist another relation between the vectors $\vecteur{z}_n$, $\vecteur{u}_n$ and $\vecteur{v}_n$, for any Hermitian matrix, $\vecteur{A}\in\mathbb{C}^{m\times m}$ \cite{MPFO13}:
{\small \begin{equation}
\label{eq:distance_real}
2\vecteur{z}_n^H\vecteur{A}^{-1}\vecteur{z}_n = \vecteur{u}_n^Tf\left(\vecteur{A}\right)^{-1}\vecteur{u}_n = \vecteur{v}_n^Tf\left(\vecteur{A}\right)^{-1}\vecteur{v}_n
\end{equation}}%
Let us apply the function $f(\cdot)$ to the equation \eqref{eq:scatter_emp_comp} (respectively $h(\cdot)$ to \eqref{eq:moy_emp_comp}), we obtain
{\footnotesize \begin{align}
&\mathbf{\widehat{M}}_{N}^{\mathbb{R}} = \dfrac{1}{2N}\sum\limits_{n=1}^N u_{2,\mathbb{R}}\left(d^2\left(\vecteur{u}_n,\mathbf{\hat{t}}_{N}^{\mathbb{R}};\mathbf{\widehat{M}}_{N}^{\mathbb{R}}\right)\right)\left(\vecteur{u}_n-\mathbf{\hat{t}}_{N}^{\mathbb{R}}\right)\left(\vecteur{u}_n-\mathbf{\hat{t}}_{N}^{\mathbb{R}}\right)^T \label{eq:scatter_emp_real}\\
&\,+ \dfrac{1}{2N}\sum\limits_{n=1}^N u_{2,\mathbb{R}}\left(d^2\left(\vecteur{v}_n,\bsym{\mathcal{P}}\mathbf{\hat{t}}_{N}^{\mathbb{R}};\mathbf{\widehat{M}}_{N}^{\mathbb{R}}\right)\right)\left(\vecteur{v}_n-\bsym{\mathcal{P}}\mathbf{\hat{t}}_{N}^{\mathbb{R}}\right)\left(\vecteur{v}_n-\bsym{\mathcal{P}}\mathbf{\hat{t}}_{N}^{\mathbb{R}}\right)^T \nonumber
\end{align}}%
{\footnotesize\vspace*{-0.5cm}
\begin{align}
{\normalsize \text{and }}\,\dfrac{1}{N}\sum\limits_{n=1}^N u_{1,\mathbb{R}}\left(d\left(\vecteur{u}_n,\mathbf{\hat{t}}_{N}^{\mathbb{R}};\mathbf{\widehat{M}}_{N}^{\mathbb{R}}\right)\right)\left(\vecteur{u}_n-\mathbf{\hat{t}}_{N}^{\mathbb{R}}\right)=\vecteur{0} \label{eq:moy_emp_real1}
\end{align}}%
where $u_{2,\mathbb{R}}(s) = u_2\left(2^{-1}s\right)$ and $u_{1,\mathbb{R}}(s) = u_1\left(2^{-1/2}s\right)$ according to \eqref{eq:distance_real}. Let $\psi_{i,\mathbb{R}}(\cdot)$ be the functions related to $u_{i,\mathbb{R}}(\cdot)$ by $\psi_{i,\mathbb{R}}(s) = su_{i,\mathbb{R}}(s),\,i=1,2$. Finally, we introduce the two following real joint $M$-estimators $\big(\mathbf{\hat{t}}_{N}^{u},\mathbf{\widehat{M}}_{N}^{u}\big)$ and $\big(\mathbf{\hat{t}}_{N}^{v},\mathbf{\widehat{M}}_{N}^{v}\big)$ respectively solution of $\mathcal{S}\mathit{ys}_N\left(\vecteur{U}_N,u_{1,\mathbb{R}},u_{2,\mathbb{R}}\right)$ and $\mathcal{S}\mathit{ys}_N\left(\vecteur{V}_N,u_{1,\mathbb{R}},u_{2,\mathbb{R}}\right)$.
From the results in the real case on the consistency \cite{Mar76,BB08}, we obtain 
{\footnotesize \begin{align}
\left\lbrace\begin{array}{l}
\big(\mathbf{\hat{t}}_{N}^{u},\mathbf{\widehat{M}}_{N}^{u}\big)\overset{\mathbb{P}}{\rightarrow}\left(\vecteur{t}_{u},\vecteur{M}_{u}\right) = \left(\vecteur{t}_{e}^{u},\sigma_{\mathbb{R}}^{-1}\bsym{\Lambda}_{\mathbb{R}}\right)\\
\big(\mathbf{\hat{t}}_{N}^{v},\mathbf{\widehat{M}}_{N}^{v}\big)\overset{\mathbb{P}}{\rightarrow}\left(\vecteur{t}_{v},\vecteur{M}_{v}\right) =\left(\vecteur{t}_{e}^{v},\sigma_{\mathbb{R}}^{-1}\bsym{\Lambda}_{\mathbb{R}}\right)
\end{array}\right.
\end{align}}%
in which $\left(\vecteur{t}_{u},\vecteur{M}_{u}\right)$ and $\left(\vecteur{t}_{v},\vecteur{M}_{v}\right)$ are solutions of  $\mathcal{S}\mathit{ys}_{\infty}\left(\vecteur{u}_1,u_{1,\mathbb{R}},u_{2,\mathbb{R}}\right)$ and $\mathcal{S}\mathit{ys}_{\infty}\left(\vecteur{v}_1,u_{1,\mathbb{R}},u_{2,\mathbb{R}}\right)$ and $\sigma_{\mathbb{R}}$ is the solution of $\mathbb{E}\left[\psi_{2,\mathbb{R}}\left(\sigma_{\mathbb{R}}\vert\vecteur{u}\vert^2\right)\right] = p = 2m$, $\vecteur{u}\sim\mathcal{ES}_p\left(\vecteur{0},\vecteur{I}_p,g_{\vecteur{z}}\right)$. Thus, we have $\vecteur{M}_{u} = \sigma_{\mathbb{R}}^{-1}\bsym{\Lambda}_{\mathbb{R}} = \vecteur{M}_{v}$.\\
Finally, by applying $\bsym{\mathcal{P}}$ to the system $\mathcal{S}\mathit{ys}_N\left(\vecteur{U}_N,u_{1,\mathbb{R}},u_{2,\mathbb{R}}\right)$, we obtain $\mathbf{\widehat{M}}_{N}^{v} = \bsym{\mathcal{P}}\mathbf{\widehat{M}}_{N}^{u}\bsym{\mathcal{P}}^T$ and $\mathbf{\hat{t}}_{N}^{v} = \bsym{\mathcal{P}}\,\mathbf{\hat{t}}_{N}^{u}$. Moreover, since $\vecteur{v}_n=\bsym{\mathcal{P}}\vecteur{u}_n,\,\forall\,n$, we have $\vecteur{M}_u = \vecteur{M}_v = \bsym{\mathcal{P}}\vecteur{M}_u\bsym{\mathcal{P}}^T$ and $\vecteur{t}_v = \bsym{\mathcal{P}}\vecteur{t}_u$.
\subsubsection{Link between asymptotic behaviors of {\small $\big(\mathbf{\hat{t}}_{N}^{u},\mathbf{\widehat{M}}_{N}^{u}\big)$}, {\small $\big(\mathbf{\hat{t}}_{N}^{v},\mathbf{\widehat{M}}_{N}^{v}\big)$} and {\small $\big(\mathbf{\hat{t}}_{N}^{\mathbb{R}},\mathbf{\widehat{M}}_{N}^{\mathbb{R}}\big)$}}
\begin{lemma}
\label{lemma1}
{\small $\mathbf{\hat{t}}_{N}^{\mathbb{R}}$} and {\small $\mathbf{\hat{t}}_{N}^{u}$} (respectively {\small $\mathbf{\widehat{M}}_{N}^{\mathbb{R}}$} and {\footnotesize $\dfrac{1}{2}\big(\mathbf{\widehat{M}}_{N}^{u}+\mathbf{\widehat{M}}_{N}^{v}\big)$}) \bm{share} the same asymptotic Gaussian law. 
\end{lemma}
\begin{proof}
See Appendix.
\end{proof}
\subsubsection{Asymptotic behavior of {\small $\big(\mathbf{\hat{t}}_N, \mathbf{\widehat{M}}_N\big)$}}
Since {\small $\vec{}\big(\mathbf{\widehat{M}}_{N}^{\mathbb{R}}\big)$} and $\mathbf{\hat{t}}_{N}^{\mathbb{R}}$ have an asymptotic Gaussian distribution according to Lemma \ref{lemma1}, and {\small $\vec{}\big(\mathbf{\widehat{M}}_N\big) = \left(\vecteur{g}_m^T\otimes\vecteur{g}_m^H\right)\vec{}\big(\mathbf{\widehat{M}}_{N}^{\mathbb{R}}\big)$} and $\mathbf{\hat{t}}_N = \vecteur{g}_m^H\mathbf{\hat{t}}_{N}^{\mathbb{R}}$ with {\small $\vecteur{g}_m = \left(\vecteur{I}_m,-j\vecteur{I}_m\right)^T$}. Consequently, $\mathbf{\hat{t}}_N$ and {\small $\vec{}\big(\mathbf{\widehat{M}}_N\big)$} have a non-circular complex Gaussian distribution \cite{Van95}. \bm{Additionally, using the same approach as in \cite{MPFO13}, we obtain}
{\footnotesize \begin{align}
&\quad\quad\sqrt{N}\vec{}\left(\mathbf{\widehat{M}}_N - \vecteur{M}_e\right) \overset{d}{\rightarrow} \vecteur{b}\sim\mathcal{G}\mathbb{C}\mathcal{N}\left(\vecteur{0}, \bsym{\Sigma}_M,\bsym{\Omega}_M\right)
\end{align}}%
Regarding the location estimate, we have
{\small 
\begin{align}
\sqrt{N}\left(\mathbf{\hat{t}}_N - \vecteur{t}_e\right) \overset{d}{\rightarrow} \mathcal{G}\mathbb{C}\mathcal{N}\left(\vecteur{0}, \bsym{\Sigma}_t,\bsym{\Omega}_t\right),\,\;\text{where}
\end{align}}
\vspace*{-0.5cm}
{\footnotesize \begin{align*}
\bsym{\Sigma}_t &= \vecteur{g}_m^HN\mathbb{E}\left[\left(\mathbf{\hat{t}}_{N}^{u} - \vecteur{t}_u\right)\left(\mathbf{\hat{t}}_{N}^{u} - \vecteur{t}_u\right)^T\right]\vecteur{g}_m \hspace*{-0.05cm}=  \dfrac{\alpha}{\beta^2}\vecteur{g}_m^H\vecteur{M}_{\mathbb{R}}\vecteur{g}_m \hspace*{-0.05cm}= \dfrac{\alpha}{\beta^2}\vecteur{M}_e
\end{align*}}%
in which
{\footnotesize \begin{align*}
\left\lbrace\begin{array}{l}
\alpha = (2m)^{-1}\mathbb{E}\left[\psi_{1,\mathbb{R}}^2(\sqrt{\sigma}\vert\vecteur{x}\vert)\right] = m^{-1}\mathbb{E}\left[\psi_{1}^2(\sqrt{\sigma}\vert\bsym{\zeta}\vert)\right] \\
\beta = \mathbb{E}\left[\left(1-(2m)^{-1}\right)u_{1}(\sqrt{\sigma}\vert\bsym{\zeta}\vert)+(2m)^{-1}\psi_{1}^{'}\left(\sqrt{\sigma}\vert\bsym{\zeta}\vert\right)\right]
\end{array}\right.
\end{align*}}%
with {\small $\vecteur{x}\sim\mathcal{ES}_p\left(\vecteur{0}_{2m},\vecteur{I}_{2m},g_{\vecteur{z}}\right)$} and {\small $\bsym{\zeta}\sim\mathbb{C}\mathcal{ES}_m\left(\vecteur{0}_{m},\vecteur{I}_{m},g_{\vecteur{z}}\right)$} hence $\vert\bsym{\zeta}\vert \overset{d}{=} 2^{-1/2}\vert\vecteur{x}\vert$. Furthermore, we have the relations $\psi_{1,\mathbb{R}}(s) = \sqrt{2}\psi_1\left(2^{-1/2}s\right)$ and $\psi_{1,\mathbb{R}}^{'}(s) = \psi_{1}^{'}\left(2^{-1/2}s\right)$. Moreover, we have
{\footnotesize \begin{align*}
\bsym{\Omega}_t &= \vecteur{g}_m^HN\mathbb{E}\left[\left(\mathbf{\hat{t}}_{N}^{\mathbb{R}} - \vecteur{t}_{\mathbb{R}}\right)\left(\mathbf{\hat{t}}_{N}^{\mathbb{R}} - \vecteur{t}_{\mathbb{R}}\right)^T\right]\vecteur{g}_m^{\ast} =  \dfrac{\alpha}{\beta^2}\vecteur{g}_m^H \vecteur{M}_{\mathbb{R}}\vecteur{g}_m^{\ast} =\vecteur{0}.
\end{align*}}%
Thus, we prove that {\small $\sqrt{N}\left(\mathbf{\hat{t}}_N - \vecteur{t}_e\right) \overset{d}{\rightarrow} \vecteur{a}\sim\mathbb{C}\mathcal{N}\left(\vecteur{0}, \bsym{\Sigma}_t\right)$}. \bm{Applying} Lemma \ref{lemma1}, we have $\bsym{\omega}\perp\bsym{\chi}$ and consequently,
{\small \begin{align*}
\left\lbrace
\begin{array}{l}
\vecteur{a}\overset{d}{=}\vecteur{g}_m^H\vecteur{M}_{\mathbb{R}}^{1/2}\vecteur{D}
^{-1}\bsym{\chi}\\
\vecteur{b}\overset{d}{=}\dfrac{1}{2}\left(\vecteur{g}_m^T\otimes\vecteur{g}_m^H\right)\left(\vecteur{I}_{p^2}+\left(\bsym{\mathcal{P}}\otimes\bsym{\mathcal{P}}\right)\right)\big(\vecteur{M}_{\mathbb{R}}^{1/2}\otimes\vecteur{M}_{\mathbb{R}}^{1/2}\big)\vecteur{A}^{-1}\bsym{\omega}
\end{array}\right.,
\end{align*}}%
thus we prove that $\vecteur{a}\perp\vecteur{b}$.
\section{Simulations}
\label{sec::simu}
In order to illustrate our theoretical results, some simulations results are presented. Two scenarios have been considered for the simulations. For $m = 3$, the true location parameter is $\vecteur{t}_e = \left(1+0.5i,2+i,3+1.5i\right)^T$ and the true scatter matrix is $\bsym{\Lambda} = \vecteur{I}_m$, due to the affine equivariance propertie of the $M$-estimators, there is no loss of generality. 
\begin{itemize}
\item Case 1 : the data are generated under a $t$-distribution with $d = 4$ degrees of freedom \cite{BA13}.
\item Case 2 : the data are generated under a $K$-distribution with shape parameter $\nu = 4$ and scale parameter $\theta = 1/\nu$ \cite{OTKP12}.
\end{itemize}
The complex joint $M$-estimators is obtained with $u(s) \triangleq u_2(s) = \dfrac{d+m}{d+s} = u_1(\sqrt{s})$ and the reweighting algorithm of \cite{KT91}, whose convergence is established.\\
The first case coincides with the MLE unlike the case 2, which is a general complex joint $M$-estimator. 
\begin{figure}[h!]
	\centering
%
\definecolor{mycolor1}{rgb}{0.00000,0.44700,0.74100}%
\definecolor{mycolor2}{rgb}{0.85000,0.32500,0.09800}%
\begin{tikzpicture}

\begin{axis}[%
width=0.72\columnwidth,
height=0.4\columnwidth,
at={(0,0)},
scale only axis,
xmode=log,
xmin=5,
xmax=100,
xminorticks=true,
xlabel style={font=\color{white!15!black}},
xlabel={\small Number of samples $N$},
ymode=log,
ymin=0.026,
ymax=3.01,
yminorticks=true,
ylabel style={font=\color{white!15!black}},
ylabel={\scriptsize $\text{Tr}\left\lbrace\textbf{MSE}\right\rbrace$ (dB)},
axis background/.style={fill=white},
legend columns=2,
legend style={at = {(0,1.04)}, legend cell align=left, anchor = south west, align=left,font=\scriptsize, draw=white!15!white}
]
\addplot [dotted, color=black,line width=1.0pt]
  table[row sep=crcr]{%
5	2.22857142857143\\
6	1.85714285714286\\
8	1.39285714285714\\
10	1.11428571428571\\
12	0.928571428571428\\
15	0.742857142857143\\
18	0.619047619047619\\
22	0.506493506493506\\
28	0.397959183673469\\
34	0.327731092436975\\
42	0.26530612244898\\
53	0.210242587601078\\
65	0.171428571428571\\
81	0.137566137566138\\
100	0.111428571428571\\
};
\addlegendentry{$\text{Tr}\left(\bsym{\Sigma}_M\right)$ case 1}

\addplot [color=mycolor1, mark=star, mark options={solid, mycolor1}]
  table[row sep=crcr]{%
5	3.00353351510446\\
6	2.29182777870947\\
8	1.53412473705703\\
10	1.20309627300127\\
12	0.988571947191243\\
15	0.77507248925973\\
18	0.637389655597625\\
22	0.527858187338509\\
28	0.40948097151813\\
34	0.334952385919148\\
42	0.267157527466804\\
53	0.211926546435785\\
65	0.17177042433827\\
81	0.138928965759236\\
100	0.112247248452662\\
};
\addlegendentry{$\text{Tr}\left\lbrace\textbf{MSE}\left(\widehat{\vecteur{M}}_N\right)\right\rbrace$ case 1}

\addplot [dash dot,color=black,line width=1.0pt]
  table[row sep=crcr]{%
5	1.33955314926295\\
6	1.11629429105246\\
8	0.837220718289345\\
10	0.669776574631476\\
12	0.55814714552623\\
15	0.446517716420984\\
18	0.372098097017487\\
22	0.304443897559762\\
28	0.239205919511241\\
34	0.196993110185728\\
42	0.159470613007494\\
53	0.126372938609712\\
65	0.103042549943304\\
81	0.0826884660038859\\
100	0.0669776574631476\\
};
\addlegendentry{$\text{Tr}\left(\bsym{\Sigma}_M\right)$ case 2}

\addplot [color=mycolor2, mark=star, mark options={solid, mycolor2}]
  table[row sep=crcr]{%
5	1.54115065411401\\
6	1.17503473746394\\
8	0.880049487494331\\
10	0.688776023995516\\
12	0.570669020244507\\
15	0.459467474078086\\
18	0.374095905747671\\
22	0.308984996725634\\
28	0.242765064219452\\
34	0.197420255783464\\
42	0.159317722686245\\
53	0.125520733843345\\
65	0.103828155438946\\
81	0.0827598561612788\\
100	0.0664890141865945\\
};
\addlegendentry{$\text{Tr}\left\lbrace\textbf{MSE}\left(\widehat{\vecteur{M}}_N\right)\right\rbrace$ case 2}


\addplot [color=black,line width=1.0pt]
  table[row sep=crcr]{%
5	0.685714285714286\\
6	0.571428571428571\\
8	0.428571428571429\\
10	0.342857142857143\\
12	0.285714285714286\\
15	0.228571428571429\\
18	0.19047619047619\\
22	0.155844155844156\\
28	0.122448979591837\\
34	0.100840336134454\\
42	0.0816326530612245\\
53	0.0646900269541779\\
65	0.0527472527472527\\
81	0.0423280423280423\\
100	0.0342857142857143\\
};
\addlegendentry{$\text{Tr}\left(\bsym{\Sigma}_t\right)$ case 1}

\addplot [color=mycolor1, mark=+, mark options={solid, mycolor1}]
  table[row sep=crcr]{%
5	0.764270972853998\\
6	0.622684480831601\\
8	0.452315585832244\\
10	0.359470270901501\\
12	0.298019575183577\\
15	0.236790547078673\\
18	0.193302772077247\\
22	0.157072541656924\\
28	0.122466408387314\\
34	0.101614611752068\\
42	0.0832696967621809\\
53	0.0653077599081408\\
65	0.0529145210154537\\
81	0.0426052714347257\\
100	0.0341595250986598\\
};
\addlegendentry{$\text{Tr}\left\lbrace\textbf{MSE}\left(\widehat{\vecteur{t}}_N\right)\right\rbrace$ case 1}

\addplot [dashed,color=black,line width=1.0pt]
  table[row sep=crcr]{%
5	0.534647858131405\\
6	0.445539881776171\\
8	0.334154911332128\\
10	0.267323929065703\\
12	0.222769940888086\\
15	0.178215952710468\\
18	0.14851329392539\\
22	0.121510876848047\\
28	0.0954728318091795\\
34	0.0786246850193243\\
42	0.063648554539453\\
53	0.0504384771822081\\
65	0.0411267583178004\\
81	0.0330029542056423\\
100	0.0267323929065703\\
};
\addlegendentry{$\text{Tr}\left(\bsym{\Sigma}_t\right)$ case 2}

\addplot [color=mycolor2, mark=+, mark options={solid, mycolor2}]
  table[row sep=crcr]{%
5	0.582614602543084\\
6	0.477986816999777\\
8	0.351303855034106\\
10	0.277830422897558\\
12	0.23218871972619\\
15	0.182891144646239\\
18	0.153297987593649\\
22	0.123636251790481\\
28	0.0964690402322863\\
34	0.0788381947473233\\
42	0.0641623922837734\\
53	0.0514240026858274\\
65	0.0416309796996717\\
81	0.0332046321391441\\
100	0.0267681902480616\\
};
\addlegendentry{$\text{Tr}\left\lbrace\textbf{MSE}\left(\widehat{\vecteur{t}}_N\right)\right\rbrace$ case 2}
\end{axis}
\end{tikzpicture}%
	\caption{Second order moment simulations}
	\label{efficiency_simu}
\end{figure}
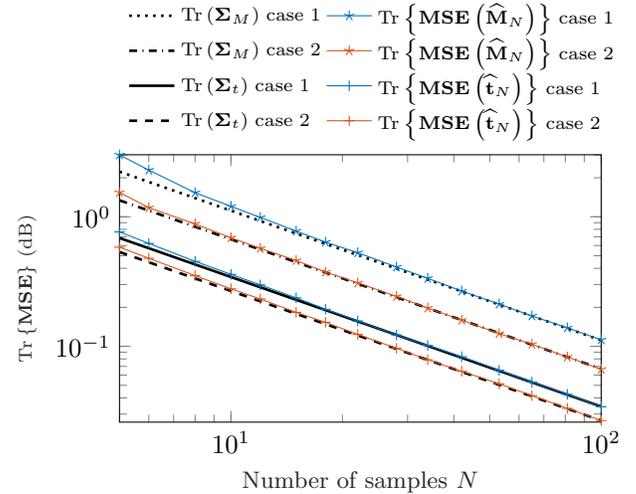

In Fig. \ref{efficiency_simu}, we plot the trace of the Mean Squared Error of the estimates of the location and the scatter matrix as well as the trace of the theoretical asymptotic covariance matrices $\bsym{\Sigma}_t$ and $\bsym{\Sigma}_M$. The results are validated since the drawn quantities are identical when $N\rightarrow\infty$. Moreover these quantities tend asymptotically to zero, which illustrate the consistency. 

\section{Conclusion}
In this letter, we established the asymptotic performance of the joint $M$-estimators for the complex multivariate location and scatter matrix. This statistical study highlights a better understanding on the performance of the $M$-estimators of \bm{the complex multivariate location and scatter matrix}. Again, \bm{the obtained results} may be used for conducting a performance analysis of adaptive processes involving non-zero mean observations.


%
{\begin{table*}[h!]\scriptsize
\captionof{table}{\bm{Equations for the proof of Lemma \ref{lemma1}}}
\begin{align}
\label{eq:dev_asymp_scatt_u}
\vecteur{I}_{p}+\Delta\vecteur{M}_u &= \dfrac{1}{N}\textstyle\sum\limits_{n=1}^N a_n\left(\vecteur{k}_n\vecteur{k}_n^T-\vecteur{k}_n\Delta\vecteur{t}_u^T\vecteur{M}_u^{-1/2} - \vecteur{M}_u^{-1/2}\Delta\vecteur{t}_u\vecteur{k}_n^T\right) + b_n\left[\vecteur{k}_n^T\Delta\vecteur{M}_u\vecteur{k}_n + 2\vecteur{k}_n^T\vecteur{M}_u^{-1/2}\Delta\vecteur{t}_u\right]\vecteur{k}_n\vecteur{k}_n^T\\
\label{eq:dev_asymp_moy_u}
\vecteur{0} &= \dfrac{1}{N}\textstyle\sum\limits_{n=1}^N c_n\left(\vecteur{k}_n-\vecteur{M}_u^{-1/2}\Delta\vecteur{t}_u\right) + \dfrac{1}{N}\sum\limits_{n=1}^N d_n\big[\dfrac{\vecteur{k}_n^T\Delta\vecteur{M}_u\vecteur{k}_n}{2\vert\vecteur{k}_n\vert} + \dfrac{\vecteur{k}_n^T\vecteur{M}_u^{-1/2}\Delta\vecteur{t}_u}{\vert\vecteur{k}_n\vert}\big]\vecteur{k}_n\\
\label{eq:A}
\vecteur{A}_N&=\vecteur{I}_{p^2}-\dfrac{1}{N}\textstyle\sum\limits_{n=1}^N b_n\left(\vecteur{k}_n\otimes\vecteur{k}_n\right)\left(\vecteur{k}_n\otimes\vecteur{k}_n\right)^T \overset{\mathbb{P}}{\rightarrow} \vecteur{A} = \vecteur{I}_{p^2}+\dfrac{1}{p(p+2)}\mathbb{E}\left[\vert\vecteur{k}_1\vert^4u_{2,\mathbb{R}}^{'}\left(\vert\vecteur{k}_1\vert^2\right)\right]\left(\vecteur{I}_{p^2}+\vecteur{K}_p+\vec{}\left(\vecteur{I}_p\right)\vec{}\left(\vecteur{I}_p\right)^T\right)\\
\vecteur{B}_N&=\dfrac{1}{N}\textstyle\sum\limits_{n=1}^N \left(\vecteur{k}_n\otimes\vecteur{I}_p\right)\left(a_n\vecteur{I}_{p} - b_n\vecteur{k}_n\vecteur{k}_n^T\right) \overset{\mathbb{P}}{\rightarrow} \mathbb{E}\left[\vert\vecteur{k}_1\vert u_{2,\mathbb{R}}\left(\vert\vecteur{k}_1\vert^2\right)\right]\mathbb{E}\left[\left(\bsym{\kappa}\otimes\vecteur{I}_p\right)\right] + \mathbb{E}\left[\vert\vecteur{k}_1\vert^3u_{2,\mathbb{R}}^{'}\left(\vert\vecteur{k}_1\vert^2\right)\right]\mathbb{E}\left[\left(\bsym{\kappa}\otimes\vecteur{I}_p\right)\bsym{\kappa}\bsym{\kappa}^T\right]\hspace*{-0.1cm} = \vecteur{0}\\
\vecteur{C}_N&=-\dfrac{1}{N}\textstyle\sum\limits_{n=1}^N d_n\dfrac{\vecteur{k}_n\vecteur{k}_n^T}{2\vert\vecteur{k}_n\vert}\left(\vecteur{I}_p\otimes\vecteur{k}_n\right)^T \hspace*{-0.1cm}\overset{\mathbb{P}}{\rightarrow} \mathbb{E}\big[u_{1,\mathbb{R}}^{'}\left(\vert\vecteur{k}_1\vert\right)\dfrac{\vecteur{k}_1\vecteur{k}_1^T}{2\vert\vecteur{k}_1\vert}\left(\vecteur{I}_p\otimes\vecteur{k}_1\right)^T\big] = \dfrac{1}{2}\mathbb{E}\left[\vert\vecteur{k}_1\vert^2u_{1,\mathbb{R}}^{'}\left(\vert\vecteur{k}_1\vert\right)\right]\mathbb{E}\left[\bsym{\kappa}\bsym{\kappa}^T\left(\vecteur{I}_p\otimes\bsym{\kappa}\right)\right]= \vecteur{0}\\
\label{eq:D}
\vecteur{D}_N &= \dfrac{1}{N}\textstyle\sum\limits_{n=1}^N c_n\vecteur{I}_p-d_n\dfrac{\vecteur{k}_n\vecteur{k}_n^T}{\vert\vecteur{k}_n\vert} \overset{\mathbb{P}}{\rightarrow} \vecteur{D} = \left(\mathbb{E}\left[u_{1,\mathbb{R}}\left(\vert\vecteur{k}_1\vert\right)\right] + \dfrac{1}{p}\mathbb{E}\left[\vert\vecteur{k}_1\vert u_{1,\mathbb{R}}^{'}\left(\vert\vecteur{k}_1\vert\right)\right]\right)\vecteur{I}_p\\
\label{eq:11}
&\hspace{-1.5cm}\left(\vecteur{I}_{p^2}+\left(\bsym{\mathcal{P}}\otimes\bsym{\mathcal{P}}\right)\right)\bsym{\omega}_N/\sqrt{N} = \left(\vecteur{A}_N + \left(\bsym{\mathcal{P}}\otimes\bsym{\mathcal{P}}\right)\vecteur{A}_N\left(\bsym{\mathcal{P}}\otimes\bsym{\mathcal{P}}\right)^T\right)\vec{}\left(\Delta\vecteur{M}_{\mathbb{R}}\right) +\left(\vecteur{I}_{p^2}+\vecteur{K}_p\right)\left(\vecteur{I}_{p^2}+\left(\bsym{\mathcal{P}}\otimes\bsym{\mathcal{P}}\right)\right)\vecteur{B}_N\vecteur{M}_{\mathbb{R}}^{-1/2}\Delta\vecteur{t}_{\mathbb{R}}\\
\mathbb{E}\left[\bsym{\omega}\bsym{\chi}^T\right] &= \mathbb{E}\left[\vert\vecteur{k}_1\vert^3u_{2,\mathbb{R}}\left(\vert\vecteur{k}_1\vert^2\right)u_{1,\mathbb{R}}\left(\vert\vecteur{k}_1\vert\right)\right]\mathbb{E}\left[\left(\bsym{\kappa}\otimes\bsym{\kappa}\right)\bsym{\kappa}^T\right]-\mathbb{E}\left[\vert\vecteur{k}_1\vert u_{2,\mathbb{R}}\left(\vert\vecteur{k}_1\vert^2\right)u_{1,\mathbb{R}}\left(\vert\vecteur{k}_1\vert\right)\right]\mathbb{E}\left[\vec{}\left(\vecteur{I}_p\right)\bsym{\kappa}^T\right] = \vecteur{0}
\label{eq:indep}
\end{align}
\end{table*}}%
\appendix[Proof of Lemma \ref{lemma1}]
\label{app::part_proof}
\subsection{Asymptotic behavior of {\small $\big(\mathbf{\hat{t}}_{N}^{u},\mathbf{\widehat{M}}_{N}^{u}\big)$} and {\small $\big(\mathbf{\hat{t}}_{N}^{v},\mathbf{\widehat{M}}_{N}^{v}\big)$}}
First of all, let us define {\small $\widehat{\vecteur{W}}_u = \vecteur{M}_u^{-1/2}\mathbf{\widehat{M}}_{N}^{u}\vecteur{M}_u^{-1/2}$} and {\small $\vecteur{k}_n = \vecteur{M}_u^{-1/2}\left(\vecteur{u}_n-\vecteur{t}_u\right)$}. Since {\small $\big(\mathbf{\hat{t}}_{N}^{u},\mathbf{\widehat{M}}_{N}^{u}\big)\overset{\mathbb{P}}{\rightarrow}\left(\vecteur{t}_{u},\vecteur{M}_{u}\right)$}, we can write for $N\rightarrow\infty$, {\small $\widehat{\vecteur{W}}_u = \vecteur{I}_{p} + \Delta\vecteur{M}_u$} and {\small $\mathbf{\hat{t}}_{N}^{u} = \vecteur{t}_u + \Delta\vecteur{t}_u$}.
Let us note {\small $a_n = u_{2,\mathbb{R}}\left(\vert\vecteur{k}_n\vert^2\right)$, $b_n = -u_{2,\mathbb{R}}^{'}\left(\vert\vecteur{k}_n\vert^2\right)$, $c_n = u_{1,\mathbb{R}}\left(\vert\vecteur{k}_n\vert\right)$} and {\small $d_n = -u_{1,\mathbb{R}}^{'}\left(\vert\vecteur{k}_n\vert\right)$} and use first order expansions for $N$ sufficiently large, then we obtain \eqref{eq:dev_asymp_scatt_u} and \eqref{eq:dev_asymp_moy_u} from $\mathcal{S}\mathit{ys}_N\left(\vecteur{U}_N,u_{1,\mathbb{R}},u_{2,\mathbb{R}}\right)$. By vectorizing \eqref{eq:dev_asymp_scatt_u} and after some calculus, we obtain
{\scriptsize \begin{align}
\label{eq:syst_dev_asymp_u}
\left\lbrace\begin{array}{l}
\bsym{\omega}_N = \vecteur{A}_N\sqrt{N}\vec{}\left(\Delta\vecteur{M}_u\right) + \left(\vecteur{I}_{p^2}+\vecteur{K}_p\right)\vecteur{B}_N\sqrt{N}\vecteur{M}_u^{-1/2}\Delta\vecteur{t}_u\\
\bsym{\chi}_N = \vecteur{C}_N\sqrt{N}\vec{}\left(\Delta\vecteur{M}_u\right) + \vecteur{D}_N\sqrt{N}\vecteur{M}_u^{-1/2}\Delta\vecteur{t}_u
\end{array}\right.
\end{align}}%
where {\scriptsize $\sqrt{N}\bsym{\omega}_N = \sum\limits_{n=1}^N a_n\left(\vecteur{k}_n\otimes\vecteur{k}_n\right)-\vec{}\left(\vecteur{I}_{p}\right)$} and {\scriptsize $\sqrt{N}\bsym{\chi}_N = \sum\limits_{n=1}^N c_n\vecteur{k}_n$}.
\begin{remark}
Note that {\footnotesize $\vecteur{k}_n\sim\mathcal{ES}_p\left(\vecteur{0},\sigma_{\mathbb{R}}\vecteur{I}_p,g_{\vecteur{z}}\right)$}. Let be {\footnotesize $\bsym{\kappa} = \dfrac{\vecteur{k}_1}{\vert\vecteur{k}_1\vert}$}, thus we have {\small $\bsym{\kappa}\perp\vert\vecteur{k}_1\vert$ and $\mathbb{E}\left[\bsym{\kappa}\right] = \vecteur{0}$, $\mathbb{E}\left[\bsym{\kappa}\bsym{\kappa}^T\right] = \dfrac{\vecteur{I}_p}{p}$}, all 3rd-order moments vanish and the only non-vanishing 4th-order moments are {\small $\mathbb{E}\left[\kappa_i^4\right] = \dfrac{3}{p(p+2)}$} and {\small $\mathbb{E}\left[\kappa_i^2\kappa_j^2\right]^{-1} = p(p+2)$} for $i\neq j$ where $\bsym{\kappa} = \left[\kappa_1,\ldots,\kappa_p\right]^T$.
\end{remark}
The SLLN yields to \eqref{eq:A}--\eqref{eq:D}.
Furthermore, since {\small $N^{-1}\sum\limits_{n=1}^N a_n\left(\vecteur{k}_n\otimes\vecteur{k}_n\right)\overset{\mathbb{P}}{\rightarrow} \vec{}\left(\vecteur{I}_p\right)$} 
and {\small $N^{-1}\sum\limits_{n=1}^N c_n\vecteur{k}_n\overset{\mathbb{P}}{\rightarrow}  \vecteur{0}_{p}$}, it yields from the central limit theorem that $\bsym{\omega}_N \overset{d}{\rightarrow} \bsym{\omega}$ and $\bsym{\chi}_N \overset{d}{\rightarrow} \bsym{\chi}$ with $\bsym{\omega}$ and $\bsym{\chi}$ zero-mean Gaussian distributed.
\bm{Applying} Slutsky's lemma \cite{MMY06}, it comes
{\footnotesize \begin{align*}
\sqrt{N}\left(\mathbf{\hat{t}}_{N}^{u}-\vecteur{t}_u\right) &\overset{d}{=} \sqrt{N}\Delta\vecteur{t}_u\overset{d}{\rightarrow}\vecteur{M}_u^{1/2}\vecteur{D}^{-1}\bsym{\chi}\\
\sqrt{N}\vec{}\left(\mathbf{\widehat{M}}_{N}^{u}-\vecteur{M}_u\right) &\overset{d}{\rightarrow}\big(\vecteur{M}_u^{1/2}\otimes\vecteur{M}_u^{1/2}\big)\vecteur{A}^{-1}\bsym{\omega}
\end{align*}}%
In the same way, we obtain
{\footnotesize \begin{align*}
&\sqrt{N}\vec{}\left(\mathbf{\widehat{M}}_{N}^{v}-\vecteur{M}_v\right)\overset{d}{\rightarrow}\left(\bsym{\mathcal{P}}\otimes\bsym{\mathcal{P}}\right)\big(\vecteur{M}_u^{1/2}\otimes\vecteur{M}_u^{1/2}\big)\vecteur{A}^{-1}\bsym{\omega}\\
&\sqrt{N}\left(\mathbf{\hat{t}}_{N}^{v}-\vecteur{t}_v\right) \overset{d}{\rightarrow}\bsym{\mathcal{P}}\vecteur{M}_u^{1/2}\vecteur{D}^{-1}\bsym{\chi}
\end{align*}}%
\subsection{Asymptotic behavior of {\small $\big(\mathbf{\hat{t}}_{N}^{\mathbb{R}},\mathbf{\widehat{M}}_{N}^{\mathbb{R}}\big)$}}
With the results of Theorem \ref{small_theorem}, the continuous mapping theorem implies
{\scriptsize \begin{align*}
&\big(\mathbf{\hat{t}}_{N}^{\mathbb{R}},\mathbf{\widehat{M}}_{N}^{\mathbb{R}}\big)=\left(h\left(\mathbf{\hat{t}}_N\right),f\left(\mathbf{\widehat{M}}_N\right)\right)\overset{\mathbb{P}}{\rightarrow}\left(h\left(\vecteur{t}_e\right),f\left(\vecteur{M}_e\right)\right)=\left(\vecteur{t}_{\mathbb{R}},\vecteur{M}_{\mathbb{R}}\right) \text{ \normalsize and } \\
&\left(h\left(\vecteur{t}_e\right),f\left(\vecteur{M}_e\right)\right) = \left(h\left(\vecteur{t}_e\right),\sigma^{-1}f\left(\bsym{\Lambda}\right)\right) = \left(h\left(\vecteur{t}_e\right),\sigma^{-1}\bsym{\Lambda}_{\mathbb{R}}\right) = \left(\vecteur{t}_u,\sigma_{\mathbb{R}}^{-1}\bsym{\Lambda}_{\mathbb{R}}\right)
\end{align*}}%
Let us define {\footnotesize $\widehat{\vecteur{W}}_{\mathbb{R}} = \vecteur{M}_{\mathbb{R}}^{-1/2}\mathbf{\widehat{M}}_{N}^{\mathbb{R}}\vecteur{M}_{\mathbb{R}}^{-1/2}$}. Since {\small $\vecteur{M}_{\mathbb{R}} = \vecteur{M}_{u}$} and {\small $\vecteur{t}_{\mathbb{R}} = \vecteur{t}_u$, $\vecteur{k}_n$} becomes {\small $\vecteur{k}_n = \vecteur{M}_{\mathbb{R}}^{-1/2}\left(\vecteur{u}_n-\vecteur{t}_{\mathbb{R}}\right)$} and satisfies {\small $\bsym{\mathcal{P}}\vecteur{k}_n = \vecteur{M}_{\mathbb{R}}^{-1/2}\left(\vecteur{v}_n-\bsym{\mathcal{P}}\vecteur{t}_{\mathbb{R}}\right)$}. For $N\rightarrow\infty$, we can write {\small $\widehat{\vecteur{W}}_{\mathbb{R}} = \vecteur{I}_{p} + \Delta\vecteur{M}_{\mathbb{R}}$} and {\small $\mathbf{\hat{t}}_{N}^{\mathbb{R}} = \vecteur{t}_{\mathbb{R}} + \Delta\vecteur{t}_{\mathbb{R}}$}.
As previously, from \eqref{eq:scatter_emp_real} we obtain \eqref{eq:11}.
Since {\footnotesize $\left(\vecteur{A}_N + \left(\bsym{\mathcal{P}}\otimes\bsym{\mathcal{P}}\right)\vecteur{A}_N\left(\bsym{\mathcal{P}}\otimes\bsym{\mathcal{P}}\right)^T\right)\overset{\mathbb{P}}{\rightarrow} 2\vecteur{A}$}, the Slutsky's lemma leads to
{\scriptsize \begin{align*}
\sqrt{N}\vec{}\left(\mathbf{\widehat{M}}_{N}^{\mathbb{R}}-\vecteur{M}_{\mathbb{R}}\right)\overset{d}{\rightarrow}&\dfrac{1}{2}\left(\vecteur{I}_{p^2}+\bsym{\mathcal{P}}\otimes\bsym{\mathcal{P}}\right)\left(\vecteur{M}_u^{1/2}\otimes\vecteur{M}_u^{1/2}\right)\vecteur{A}^{-1}\bsym{\omega}
\end{align*}}%
Similarly, from \eqref{eq:moy_emp_real1} we obtain
{\footnotesize \begin{align*}
\bsym{\chi}_N = \vecteur{C}_N\sqrt{N}\vec{}\left(\Delta\vecteur{M}_{\mathbb{R}}\right) + \vecteur{D}_N\sqrt{N}\vecteur{M}_{\mathbb{R}}^{-1/2}\Delta\vecteur{t}_{\mathbb{R}}
\end{align*}}%
and thus {\small $\sqrt{N}\left(\mathbf{\hat{t}}_{N}^{\mathbb{R}}-\vecteur{t}_{\mathbb{R}}\right)\overset{d}{\rightarrow}\vecteur{M}_{\mathbb{R}}^{1/2}\vecteur{D}^{-1}\bsym{\chi}$}, which means that $\mathbf{\hat{t}}_{N}^{\mathbb{R}}$ and $\mathbf{\hat{t}}_{N}^{u}$ have the same asymptotic distribution.\\
 
\bm{Lastly}, we also introduce {\footnotesize $\bsym{\xi}_n = \begin{bmatrix}
a_n\left(\vecteur{k}_n\otimes\vecteur{k}_n\right)-\vec{}\left(\vecteur{I}_{p}\right) \\ 
c_n\vecteur{k}_n
\end{bmatrix}$}, which are zero mean and i.i.d., then 
{\footnotesize \begin{align*}
\begin{pmatrix}
\bsym{\omega}_N \\ 
\bsym{\chi}_N
\end{pmatrix} = \dfrac{1}{\sqrt{N}}\sum\limits_{n=1}^N \bsym{\xi}_n\overset{d}{\rightarrow} \bsym{\xi}=\begin{pmatrix}
\bsym{\omega} \\ 
\bsym{\chi}
\end{pmatrix}\sim\mathcal{N}\left(\vecteur{0},\text{Var}\left(\bsym{\xi}_1\right)\right)
\end{align*}}%
Furthermore, since we have \eqref{eq:indep}, we obtain $\bsym{\omega}\perp\bsym{\chi}$.


\ifCLASSOPTIONcaptionsoff
  \newpage
\fi

\bibliographystyle{IEEEtran}
\bibliography{C:/Users/meriaux_bru/Desktop/Articles_conf_papiers_journal/Bruno}

\end{document}